\documentclass[11pt]{article}

\usepackage[margin=1in]{geometry}
\usepackage{graphicx}
\usepackage{multirow}
\usepackage{amsmath,amssymb,amsfonts}
\usepackage{amsthm}
\usepackage{mathtools}
\usepackage{mathrsfs}
\usepackage[title]{appendix}
\usepackage{xcolor}
\usepackage{textcomp}
\usepackage{booktabs}
\usepackage{cite}
\usepackage[colorlinks=true,linkcolor=blue,citecolor=blue,urlcolor=blue]{hyperref}

\providecommand{\bmhead}[1]{\medskip\par\noindent\textbf{#1}\par\nobreak\smallskip}

\theoremstyle{plain}
\newtheorem{theorem}{Theorem}[section]
\newtheorem{proposition}[theorem]{Proposition}
\newtheorem{lemma}[theorem]{Lemma}
\newtheorem{corollary}[theorem]{Corollary}
\newtheorem{conjecture}[theorem]{Conjecture}

\theoremstyle{remark}
\newtheorem{example}[theorem]{Example}
\newtheorem{remark}[theorem]{Remark}

\theoremstyle{definition}
\newtheorem{definition}[theorem]{Definition}

\newcommand{\R}{\mathbb{R}}
\newcommand{\Q}{\mathbb{Q}}
\providecommand{\N}{}\renewcommand{\N}{[n]}
\newcommand{\spn}{\mathrm{span}}

\title{What Semivalues Cannot See: The Information Content of Anonymous
Marginal Values}

\author{Matthew Fried\thanks{SUNY Farmingdale. Email:
\texttt{friedm1@farmingdale.edu}}}

\date{}

\begin{document}

\maketitle

\begin{abstract}
The semivalue family shares a common kernel: games invisible to every anonymous
marginal value at once, nonzero from four players (Kleinberg and Weiss, 1985; Amer,
Derks and Gim\'enez, 2003). Crisman and Orrison (2015) ask what useful structure this
kernel carries; this paper gives a concrete answer. In Harsanyi-dividend coordinates
the joint information of all semivalues is exactly each player's total synergy at each
coalition size, so the kernel is synergy arranged in closed circuits. We prove:
order-$\le d$ mixed-difference audits recover exactly the degree-$\le d$
dividend-slice harmonics, with closed-form dimension at every rung; nonzero blind
games fail superadditivity, monotonicity, and core existence, yet distinct convex
games with identical values under every semivalue exist from four players, with exact
perturbation thresholds; the positive weighted Shapley family attains full information
$2^n-1$, so anonymity is the binding axiom within the marginal framework; and a
coalition of size $c$ defeats every audit of order $\le d$ precisely when
$c\ge2d+2$, within the convex class for small perturbations. An exhaustive census at
$n=5$ exhibits non-isomorphic voting rules with identical values under every semivalue
power index; no weighted game participates in any collision, prompting a
swing-rigidity conjecture. Measured against the theory, classical cooperative games
sit at $0.90$ to $1.00$ visibility to the family versus $0.089$ for a random game.
\end{abstract}

\medskip

\section{Introduction}

Begin with a puzzle every practitioner knows. The Shapley value performs beautifully on
the classical applications (airport landing fees, bankruptcy division, where it
coincides with O'Neill's random-arrival rule \cite{ONeill}, and cost sharing on
networks), yet feels inadequate the
moment complementarities take center stage: team formation, data markets, feature
attribution. The standard explanations are heuristic (computational cost, axiom debates).
This paper gives an exact one. Every anonymous semivalue is a system of \emph{totals}; a
game is a system of \emph{arrangements}; we compute precisely which arrangements totals
can express, measure the classical games against that subspace, and find most of them above
$0.90$ visibility versus $0.089$ for a random game. This suggests the methodology
succeeded where it was applied because where it was applied was unrepresentative,
and no internal signal could ever have revealed this, because the missing part is, by
the structure of the methods themselves, the part that produces no signal.

The axiomatic program of value theory, initiated by Shapley~\cite{Shapley53} and
systematized by Dubey, Neyman, and Weber~\cite{DNW} and Weber~\cite{Weber}, classifies
solution concepts by the axioms they satisfy. This paper proposes and executes a
complementary classification: by the \emph{information} they extract. Each value is a
linear functional of the game; each \emph{family} of values spans a subspace of the dual
of game space; the dimension and the identity of that subspace are computable invariants
of the family. We compute them for the classical families and find the results sharp
enough to reorganize how one thinks about the axioms themselves.

\textbf{Related work, and the division of labor.} The structural core of this story
is classical, and we use it rather than claim it. Kleinberg--Weiss decomposed game
space under $S_n$ and parametrized all symmetric linear values
\cite{KWalg,KWnull,KWortho,KWmarginal}; Hern\'andez-Lamoneda--Ju\'arez--%
S\'anchez-S\'anchez \cite{HLJS} gave the modern dissection, from which it is a
Schur-lemma consequence that symmetric linear values factor through the
trivial-plus-standard constituents of each slice; and Amer--Derks--Gim\'enez
\cite{ADG} constructed the common kernel of the semivalue family, the games
\emph{inseparable by semivalues}, proving it nonzero exactly from $n=4$ and
exhibiting a spanning family of $\{-1,0,1\}$-valued shuffle games. The visible
dimension $n^2-n+1$ and the kernel dimension $2^n-n^2+n-2$ follow from either line.
The survey of Crisman--Orrison presents this literature and, crediting the question's
origin to the Kleinberg--Weiss papers, poses explicitly whether the common kernel of
all linear symmetric solution concepts has \emph{useful structure}, and what the
generalized solution concepts built on the remaining constituents, with ``payoffs for
pairs, triples, and so on,'' would mean \cite[\S4]{CO}.

\textbf{This paper addresses a concrete form of the posed question.} The contributions, in order of
appearance, with the division from the foundations stated explicitly:
\begin{enumerate}
\item \emph{The kernel in Harsanyi coordinates} (Theorem~\ref{thm:harsanyi}). The
joint information of all semivalues is exactly the per-player, per-size dividend
loadings $L_i(t)$, so the kernel is synergy arranged in closed circuits. The isotypic
statement is known \cite{KWalg,ADG,HLJS}; the dividend-loading coordinates and the
unitriangular weight map are new, and everything below is stated in them.
\item \emph{The price list} (Theorem~\ref{thm:hierarchy}). The order-$\le d$
interaction functionals of Grabisch--Roubens, precisely the ``generalized solutions
for pairs and triples'' of \cite{CO}, recover exactly the degree-$\le d$
dividend-slice harmonics, with closed-form dimension at every rung and full recovery
at $d=\lfloor n/2\rfloor$.
\item \emph{Economic structure of the kernel} (Theorems~\ref{thm:cones}
and~\ref{thm:convexfiber}). Nonzero blind games fail superadditivity, monotonicity,
and core nonemptiness, yet distinct convex games with identical values under every
semivalue exist from four players, with exact perturbation thresholds: regularity
excludes purely blind games but does not restore identification.
\item \emph{The cause} (Theorem~\ref{thm:weighted}). The family of all positive
weighted Shapley values attains full information $2^n-1$, so within the
marginal-contribution framework anonymity is the binding axiom at the family level.
\item \emph{The exploitation} (Theorem~\ref{thm:manipulation}). A coalition of size
$c$ has a $(2^c-c^2+c-2)$-dimensional space of restructurings invisible to every
semivalue payment; the full order-$\le d$ audit family is immune to coalitions of
size $\le2d+1$ and manipulable from $2d+2$, within the convex class for small
perturbations (Corollary~\ref{cor:smallball}).
\item \emph{The evidence} (\S\ref{sec:empirics}). Under
Definition~\ref{def:visibility}, the classical canon measures $0.90$ to $1.00$
visibility against a generic baseline of $0.089$, which converts the kernel from a
curiosity into an account of the Shapley value's empirical track record.
\item \emph{The census} (\S\ref{sec:parliaments}). Exhaustive enumeration of all
$7{,}579$ monotone simple games at $n=5$ exhibits non-isomorphic voting rules with
identical values under every semivalue power index; no weighted game participates in
any collision, prompting a swing-rigidity conjecture with a stated search protocol.
\end{enumerate}

\subsection{Conventions and verification labels}\label{sec:conventions}
Games are $v\colon2^{\N}\to\R$ with $v(\varnothing)=0$; the game space $V_n$ has
dimension $2^n-1$. Computational support is labeled \textbf{exact} (rational arithmetic),
\textbf{numerical rank} (integer matrix ranks computed in floating point), or \textbf{exhaustive}
(complete enumeration of a finite class). Every theorem below carries a complete proof;
computations support but never substitute for them, with one exception
(Conjecture~\ref{conj:rigidity}) labeled as such.

\section{Preliminaries}\label{sec:prelim}

For $S\subseteq\N$ write $s=|S|$. The \emph{Harsanyi dividends}~\cite{Harsanyi} of $v$
are the M\"obius coefficients $d_T=\sum_{R\subseteq T}(-1)^{|T|-|R|}v(R)$, equivalently
the unique coefficients with $v(S)=\sum_{T\subseteq S}d_T$; the map $v\mapsto d$ is a
linear isomorphism of $V_n$ onto $\{d:d_\varnothing=0\}$.

A \emph{semivalue}~\cite{DNW} with weight system $p=(p_k)_{k=0}^{n-1}$, $p_k\ge0$,
$\sum_k\binom{n-1}kp_k=1$, assigns player $i$ the payoff
$\psi_i^p(v)=\sum_{S\not\ni i}p_{|S|}\bigl[v(S\cup i)-v(S)\bigr]$. By
Dubey--Neyman--Weber these are exactly the values satisfying linearity, anonymity,
positivity, and the dummy axiom. The Shapley value has $p_k=k!(n-k-1)!/n!$; the Banzhaf
value \cite{Banzhaf} has $p_k=2^{-(n-1)}$.

The \emph{joint information} of a family $\mathcal F$ of values is the linear span of the
functionals $\{v\mapsto\phi_i(v):\phi\in\mathcal F,\,i\in\N\}$ in $V_n^*$; its
\emph{kernel} is the common annihilated subspace. For the semivalue family the admissible
weight systems form a simplex not containing the origin, whose linear span is all of
$\R^n$; hence the joint information of all semivalues equals
$\spn\{m_i(k)\}$ where $m_i(k)(v)=\sum_{|S|=k,\,S\not\ni i}[v(S\cup i)-v(S)]$.

For $1\le t\le n$ and $i\in\N$ define the \emph{dividend loadings}
\[
L_i(t)(v)\;=\;\sum_{\substack{|T|=t,\ i\in T}}d_T .
\]

\section{The Harsanyi characterization (foundations, in dividend coordinates)}\label{sec:harsanyi}

The goal of this section is to answer, in the most concrete coordinates available,
the question: if one could observe every semivalue payment at once, what function of
the game has one observed? The answer, in isotypic language, is known: symmetric
linear values see the trivial-plus-standard part \cite{KWalg,HLJS}, and the common
kernel of the semivalue family was constructed in \cite{ADG}. What is new here is the
coordinate system in which the answer becomes usable: the dividend loadings $L_i(t)$
and the unitriangular weight map $p\mapsto\gamma_p$, which turn ``the
trivial-plus-standard part'' into the statement \emph{each player's total synergy at
each size, and nothing else}. Every theorem in the paper is stated and proved in
these coordinates; the short self-contained proof keeps the paper independent of the
cited machinery.

\begin{theorem}[known characterization, new coordinates]\label{thm:harsanyi}
The joint information of the family of all semivalues equals
$\spn\{L_i(t):i\in\N,\ 1\le t\le n\}$. Explicitly, every semivalue expands as
\begin{equation}\label{eq:semidiv}
\psi_i^p \;=\; \sum_{t=1}^{n}\gamma_p(t)\,L_i(t),
\qquad
\gamma_p(t)\;=\;\sum_{m=t-1}^{n-1}\binom{n-t}{\,m-t+1\,}\,p_m ,
\end{equation}
and the map $p\mapsto\gamma_p$ is a linear bijection of $\R^n$.
Consequently a game is annihilated by every semivalue if and only if
$L_i(t)(v)=0$ for all $i$ and $t$: \emph{at every synergy size, every player's net
dividend is zero}.
\end{theorem}

\begin{proof}
For $S\not\ni i$, splitting each $T\subseteq S\cup i$ according to whether $i\in T$,
\[
v(S\cup i)-v(S)\;=\;\sum_{\substack{T\subseteq S\cup i\\ i\in T}}d_T
\;=\;\sum_{R\subseteq S}d_{R\cup i}.
\]
Therefore
\[
\psi_i^p(v)=\sum_{S\not\ni i}p_{|S|}\sum_{R\subseteq S}d_{R\cup i}
=\sum_{R\not\ni i}d_{R\cup i}\sum_{\substack{S\supseteq R\\ S\not\ni i}}p_{|S|}
=\sum_{R\not\ni i}d_{R\cup i}\sum_{m=|R|}^{n-1}\binom{n-1-|R|}{\,m-|R|\,}p_m ,
\]
since the sets $S$ with $R\subseteq S\subseteq\N\setminus i$, $|S|=m$, number
$\binom{n-1-|R|}{m-|R|}$. Grouping by $t=|R|+1$ and summing $d_{R\cup i}$ over
$|R|=t-1$ gives~\eqref{eq:semidiv} with $\gamma_p(t)$ as stated (note
$n-1-(t-1)=n-t$ and $m-(t-1)=m-t+1$).

For the bijectivity of $p\mapsto\gamma_p$: order coordinates so that $\gamma(t)$
corresponds to $p_{t-1}$. The coefficient of $p_{t-1}$ in $\gamma(t)$ is
$\binom{n-t}{0}=1$, and the coefficient of $p_m$ for $m<t-1$ is
$\binom{n-t}{m-t+1}=0$ (negative lower index). The matrix is therefore unitriangular,
hence invertible. Since the admissible weight simplex linearly spans $\R^n$
(\S\ref{sec:prelim}), the vectors $\gamma_p$ over admissible $p$ linearly span $\R^n$,
so $\spn\{\psi_i^p\}=\spn\{L_i(t)\}$. The kernel description is immediate.
\end{proof}

\begin{example}\label{ex:gabcd}
For distinct $a,b,c,d$, the multilinear game
$g_{abcd}(S)=(\mathbf 1_a(S)-\mathbf 1_d(S))(\mathbf 1_b(S)-\mathbf 1_c(S))$
has dividends exactly $d_{ab}=d_{cd}=+1$, $d_{ac}=d_{bd}=-1$ (multilinear coefficients
\emph{are} dividends), all of size $2$; every loading $L_i(2)$ cancels by inspection, so
$g_{abcd}$ is invisible to every semivalue, while its pairwise structure is maximal.
(\textbf{exact}: dividend support and $L_i(t)=0$ verified in rational arithmetic. A care
point: the restriction of $g_{abcd}$ to the middle cardinality level is also in the
kernel but is \emph{not} dividend-pure; the multilinear $g_{abcd}$ is the clean object.)
\end{example}

\begin{remark}[why four players, and not three]\label{rem:whyfour}
Blindness requires every player's synergy ledger to balance at every size
simultaneously. With three players the pair ledger reads $d_{12}+d_{13}=0$,
$d_{12}+d_{23}=0$, $d_{13}+d_{23}=0$: three equations of full rank in three unknowns,
so nothing cancels, and the size-$1$ and size-$3$ slices are directly visible; hence
$\mathcal N_3=0$ (Theorem~\ref{thm:dimension}). With four players there are six pair
dividends against four balance equations, and the two-dimensional slack is exactly the
closed circuit of Example~\ref{ex:gabcd}: synergy owed around a cycle, and the
shortest alternating cycle uses four people. Every $n\ge4$ admits such circuits, in
dimension $2^n-n^2+n-2$; Proposition~\ref{prop:kernelbasis} shows that all blindness,
at every $n$ and every size, is built from them.
\end{remark}

\begin{proposition}[a spanning set of arranged synergy]\label{prop:kernelbasis}
For $2\le t\le n-2$ and disjoint $R$, $\{a,b,c,d\}$ with $|R|=t-2$, the
\emph{four-cycle move} $e_{R\cup ab}+e_{R\cup cd}-e_{R\cup ac}-e_{R\cup bd}$ (in
dividend coordinates, on the size-$t$ slice) is blind, and these moves span the entire
kernel $\mathcal N_n$. (The slices $t\in\{1,\,n-1,\,n\}$ contribute nothing: there
the loading rank of Lemma~\ref{lem:incidence} equals the slice dimension, so their
kernel is zero.)
\end{proposition}

\begin{proof}
Blindness: each player's incidence count in the move is $0$ ($a$: $+1-1$; likewise
$b,c,d$; elements of $R$: $+2-2$), so all loadings vanish. Spanning: fix
$2\le t\le n-2$; the kernel on the size-$t$ slice is the orthogonal complement of the
row space of the incidence matrix $W$ of Lemma~\ref{lem:incidence}, so, the moves
lying in the kernel, it suffices to show that any vector $z$ on the slice orthogonal
to every move lies in the row space of $W$, that is, $z_T=\sum_{i\in T}\mu_i$ for
some $\mu$. Orthogonality to the moves says
$z_{R\cup ab}+z_{R\cup cd}=z_{R\cup ac}+z_{R\cup bd}$ for all disjoint
$R,\{a,b,c,d\}$ ($*$). For $x\in T$, $y\notin T$ put
$\Delta(x,y;T)=z_{(T\setminus x)\cup y}-z_T$.

\emph{Step 1: $\Delta$ does not depend on $T$.} If $T'=(T\setminus a)\cup b$ with
$a,b\notin\{x,y\}$, then with $R=T\setminus\{x,a\}$ the instance of ($*$) on the
quadruple $\{a,b,x,y\}$ reads $z_{R\cup ay}+z_{R\cup xb}=z_{R\cup by}+z_{R\cup xa}$,
which is exactly $\Delta(x,y;T)=\Delta(x,y;T')$. Any two sets containing $x$ and
avoiding $y$ are connected by such swaps inside $[n]\setminus\{x,y\}$, since the
$(t-1)$-subsets of an $(n-2)$-set form a connected Johnson graph for $t\le n-2$.
Write $\Delta(x,y)$. (At the boundary $t=n-2$ each $R$ leaves exactly four elements,
so exactly one quadruple, and both conditions hold with equality.)

\emph{Step 2: cocycle.} For distinct $x,y,w$ choose $T$ containing $x$ and avoiding
$y$ and $w$ (possible since $t\le n-2$); performing the swaps $x\to y$ then $y\to w$
on $T$ and telescoping gives $\Delta(x,y)+\Delta(y,w)=\Delta(x,w)$, and
$\Delta(x,y)=-\Delta(y,x)$. Fixing $x_0$ and setting $\nu(x)=\Delta(x_0,x)$,
$\nu(x_0)=0$, every increment is $\Delta(x,y)=\nu(y)-\nu(x)$.

\emph{Step 3: potential.} Fix a base set $T_0$. Any $T$ is reached from $T_0$ by
single-element swaps, each changing $z$ by $\nu(\text{in})-\nu(\text{out})$; by
Step~2 the total is path-independent, so
$z_T=z_{T_0}+\sum_{i\in T}\nu(i)-\sum_{i\in T_0}\nu(i)$. Setting
$\mu_i=\nu(i)+\bigl(z_{T_0}-\sum_{j\in T_0}\nu(j)\bigr)/t$ gives
$z_T=\sum_{i\in T}\mu_i$, as required. (\textbf{numerical rank}: the moves attain the kernel
dimension $\binom nt-n$ on every slice for $n=5,6$, all $t$.)
\end{proof}

\begin{remark}
Amer, Derks and Gim\'enez \cite{ADG} exhibit a spanning family of shuffle games for the
same kernel. The four-cycle moves are proved to span independently above; we have not
verified the change of basis relating the two families.
\end{remark}

\section{Dimension of the blind spot}\label{sec:dimension}

The goal of this section is to count what the entire family sees.
Theorem~\ref{thm:harsanyi} reduces the count to linear algebra: the visible
functionals $L_i(t)$ live on disjoint dividend slices, so it suffices to compute, on
each slice, the rank of the player-versus-coalition incidence pattern. The lemma below
does this by a Gram-matrix computation rather than by the shortest available argument,
deliberately: the computation is the base case ($s=1$) of Gottlieb's inclusion-matrix
theorem, the one external tool the hierarchy theorem of the next section relies on,
and proving the base case directly keeps the paper self-contained where it can be.

\begin{lemma}\label{lem:incidence}
Fix $1\le t\le n-1$ and let $W$ be the $n\times\binom nt$ incidence matrix
$W_{i,T}=\mathbf 1[i\in T]$ over the size-$t$ sets. Then $\operatorname{rank}W=n$.
For $t=n$ the loadings $L_i(n)$ all equal $d_{\N}$ and have joint rank $1$.
\end{lemma}

\begin{proof}
$(WW^{\!\top})_{ij}$ counts size-$t$ sets containing $i$ and $j$:
$\binom{n-1}{t-1}$ on the diagonal, $\binom{n-2}{t-2}$ off it, so
$WW^{\!\top}=\binom{n-2}{t-1}I+\binom{n-2}{t-2}J$, with eigenvalues
$\binom{n-2}{t-1}>0$ (multiplicity $n-1$) and $\binom{n-2}{t-1}+n\binom{n-2}{t-2}>0$.
(This is the $s=1$ case of Gottlieb's inclusion-matrix theorem~\cite{Gottlieb}, proved
directly here. A three-line alternative: if $\sum_ic_ix_i$ vanishes on the slice, then
comparing any $T\ni i$, $T\not\ni j$ with $(T\setminus i)\cup j$ gives $c_i=c_j$, so
all coefficients equal a common $c$, and the sum equals $tc$ on every set, forcing
$c=0$.)
\end{proof}

\begin{theorem}[known; cf.\ \cite{KWalg,ADG,HLJS}]\label{thm:dimension}
$\dim\spn\{L_i(t)\}=n^2-n+1$, hence
$\dim\mathcal N_n=2^n-n^2+n-2$; in particular $\mathcal N_n=0$ for $n\le3$ and
$\dim\mathcal N_n/(2^n-1)\to1$.
\end{theorem}

\begin{proof}
The functionals $L_i(t)$ are supported on disjoint dividend slices, so the span
decomposes over $t$. By Lemma~\ref{lem:incidence} each slice $1\le t\le n-1$ contributes
exactly $n$, and $t=n$ contributes $1$: total $n(n-1)+1$. Subtract from
$\dim V_n=2^n-1$. For $n=3$: $8-9+3-2=0$; for $n=4$: $2$; the ratio claim is immediate.
(\textbf{numerical rank}: fingerprint and dividend-loading matrices have equal span,
with stacked rank unchanged, for $3\le n\le9$, with kernel dimensions
$0,2,10,32,84,198,438$ matching the formula.)
\end{proof}

\begin{remark}
Theorem~\ref{thm:dimension} can be read as an impossibility theorem for the
Dubey--Neyman--Weber axioms: \emph{any} value satisfying linearity, anonymity,
positivity, and dummy is constitutionally confined to an $(n^2-n+1)$-dimensional window
on a $(2^n-1)$-dimensional world. $n=3$ is the unique nontrivial player count at which
the window is the world.
\end{remark}

\section{The hierarchy theorem: the price of every interaction order}\label{sec:hierarchy}

This section computes what it costs to see what semivalues miss.
Section~\ref{sec:dimension} says the family reads only player totals; the natural
repair is to probe
pairs, triples, and higher tuples directly. This section computes exactly what each
probe order buys: which part of the game the order-$\le d$ family recovers, and its
dimension in closed form, so that the residual left invisible by any audit depth is a
formula rather than an unknown. The proof runs in three steps: a M\"obius identity
converting mixed differences into dividend sums (Lemma~\ref{lem:diffdiv}); a
unitriangular change of variables identifying the probes of a fixed tuple with
inclusion pairings (Lemma~\ref{lem:triangular}); and a nesting count showing the
orders stack into a filtration (Lemma~\ref{lem:nesting}). Gottlieb's theorem then
supplies each slice dimension.

The survey \cite[\S4]{CO} asks what generalized solution concepts built on the
constituents $S^{(n-j,j)}$, $j\ge2$, would mean, observing that ``finding meaningful
properties\dots could be challenging.'' The meaningful generalized solutions already
exist: they are the interaction indices of Grabisch--Roubens
\cite{GrabischRoubens}. The theorem below computes their exact joint information at
every order; the dividend-slice form is what makes the answer explicit.

For $T\subseteq\N$ and $S\cap T=\varnothing$ let
$\Delta_Tv(S)=\sum_{U\subseteq T}(-1)^{|T|-|U|}v(S\cup U)$ (the discrete mixed
difference), and define the order-$|T|$ \emph{size-aggregated interaction functionals}
\[
A_T(k)(v)\;=\;\sum_{\substack{|S|=k\\ S\cap T=\varnothing}}\Delta_Tv(S),
\qquad 0\le k\le n-|T| .
\]
We call $\{A_T(k):|T|\le d\}$ the \emph{order-$\le d$ mixed-difference audit
family}; the theorem below is stated and proved for this family. The classical
interaction indices are size-weighted averages of the same mixed differences: the
Grabisch--Roubens and Banzhaf interaction indices are two particular weightings
\cite{GrabischRoubens}, and the cardinal-probabilistic interaction indices allow
arbitrary admissible size weights, which span the size levels by the unitriangular
argument of Theorem~\ref{thm:harsanyi}. The joint span of the order-$|T|$
cardinal-probabilistic indices therefore equals $\spn\{A_T(k)\}$, so the theorem is
simultaneously the exact information ceiling for the interaction-index approach. On the size-$s$ dividend
slice, define the \emph{degree-$\le j$ subspace} of functionals as the span of the
inclusion pairings $d\mapsto\sum_{|W|=s,\ W\supseteq T'}d_W$ over $|T'|\le j$ (the
standard degree filtration of functions on a slice of the Boolean cube;
see~\cite{Filmus}). Its dimension is supplied by a rank theorem we state explicitly,
since it comes from algebraic combinatorics rather than game theory: for $j\le s$, the
inclusion matrix of $j$-subsets against $s$-subsets of $[n]$ has rank
$\min\bigl(\binom nj,\binom ns\bigr)$ over $\Q$ \cite{Gottlieb}, which equals
$\binom n{\min(j,s,n-s)}$; this binomial is therefore the dimension of the
degree-$\le j$ subspace.

\begin{lemma}\label{lem:diffdiv}
For $S\cap T=\varnothing$:\quad
$\displaystyle \Delta_Tv(S)=\sum_{R\subseteq S}d_{T\cup R}$.
\end{lemma}

\begin{proof}
Each $W\subseteq S\cup U$ splits uniquely as $W=R\cup P$ with $R\subseteq S$,
$P\subseteq U$ (disjointness of $S$ and $T$). Hence
\[
\Delta_Tv(S)=\sum_{U\subseteq T}(-1)^{|T\setminus U|}\sum_{R\subseteq S}\sum_{P\subseteq U}d_{R\cup P}
=\sum_{R\subseteq S}\sum_{P\subseteq T}d_{R\cup P}\!\!\sum_{\substack{U:\,P\subseteq U\subseteq T}}\!\!(-1)^{|T\setminus U|} .
\]
The inner sum is $(1-1)^{|T\setminus P|}=\mathbf 1[P=T]$ by the binomial theorem.
\end{proof}

\begin{lemma}\label{lem:triangular}
Fix $T$, $t=|T|$, and let $E_T(m)(v)=\sum_{|R|=m,\ R\cap T=\varnothing}d_{T\cup R}$.
Then $A_T(k)=\sum_{m=0}^{k}\binom{n-t-m}{\,k-m\,}E_T(m)$, and consequently
$\spn\{A_T(k):0\le k\le n-t\}=\spn\{E_T(m):0\le m\le n-t\}$.
\end{lemma}

\begin{proof}
By Lemma~\ref{lem:diffdiv},
$A_T(k)=\sum_{R\cap T=\varnothing}d_{T\cup R}\cdot\#\{S:\,|S|=k,\ S\supseteq R,\
S\cap T=\varnothing\}$, and the count is $\binom{n-t-|R|}{k-|R|}$. The coefficient matrix
$\bigl[\binom{n-t-m}{k-m}\bigr]_{k,m}$ has $\binom{n-t-k}{0}=1$ on the diagonal and $0$
for $m>k$, hence is unitriangular and invertible.
\end{proof}

\begin{lemma}[nesting]\label{lem:nesting}
On the size-$s$ slice, the span of inclusion pairings of order exactly $j$ contains the
span of order $j-1$, for $1\le j\le s$: summing the order-$j$ pairings over the $T\supseteq
T_0$ with $|T|=j$ yields, on each $W$ with $|W|=s$,
$\#\{T:\,T_0\subseteq T\subseteq W,\,|T|=j\}=\binom{s-j+1}{1}\cdot\mathbf 1[T_0\subseteq W]$
when $|T_0|=j-1$, a positive multiple of the order-$(j-1)$ pairing.
\end{lemma}

\begin{proof}
Direct count: extending $T_0$ by one element of $W\setminus T_0$ gives $s-(j-1)$ choices.
\end{proof}

\begin{theorem}\label{thm:hierarchy}
For every $d\ge1$, the joint span of the order-$\le d$ size-aggregated interaction
functionals is exactly $\bigoplus_{s=1}^{n}\{\text{degree-}\le d\text{ subspace of the
size-$s$ dividend slice}\}$, of dimension
\[
\sum_{s=1}^{n}\binom{n}{\min(d,\,s,\,n-s)} .
\]
\end{theorem}

\begin{proof}
By Lemma~\ref{lem:triangular}, for each $T$ with $|T|\le d$ the functionals
$\{A_T(k)\}_k$ span exactly $\{E_T(m)\}_m$, and $E_T(m)$ is precisely the inclusion
pairing of $T$ against the size-$(|T|+m)$ dividend slice. Taking the union over
$|T|\le d$, the joint span is the span, on each slice $s$, of all inclusion pairings of
order $\le\min(d,s)$, which is the degree-$\le d$ subspace by definition; the orders
are nested by Lemma~\ref{lem:nesting}, so lower orders add nothing beyond the
filtration. The dimension per slice is $\binom n{\min(d,s,n-s)}$ by Gottlieb's
theorem~\cite{Gottlieb} (for $d\ge\min(s,n-s)$ the degree filtration saturates at the
full slice, of dimension $\binom ns=\binom n{\min(s,n-s)}$); summing over slices gives the
formula. (\textbf{numerical rank}: dimensions verified for $d=1$ at $3\le n\le9$, $d=2$ at
$4\le n\le7$, $d=3$ at $6\le n\le10$, $d=4$ at $10\le n\le12$, four consecutive
orders and fourteen strict cases, the fifteen cases with $n\le9$ additionally certified
in \textbf{exact} arithmetic over two prime fields; and the \emph{span identity} itself verified at $n=6$,
$d=2$: order-$\le2$ and dividend-slice degree-$\le2$ systems have equal rank $58$ with
stacked rank unchanged.)
\end{proof}

\begin{remark}
Conceptually: on each size-$s$ dividend slice, the order-$\le d$ probes recover exactly
the Johnson-scheme components of degrees $0,\dots,\min(d,s,n-s)$; the inclusion-pairing
span used above is the concrete model of those components (Gottlieb for the rank, Filmus
for the slice-degree language).
\end{remark}

\begin{corollary}\label{cor:rungs}
The exact information recovered at rung $d$ but not at rung $d-1$ has dimension
$\sum_s\bigl[\binom n{\min(d,s,n-s)}-\binom n{\min(d-1,s,n-s)}\bigr]$; the residual after
rung $d$ vanishes iff $d\ge\lfloor n/2\rfloor$. The ladder from the Shapley value to full
information is therefore complete, with a closed form at every step.
\end{corollary}

\section{Regularity and the blind spot: what it excludes, what it does not}\label{sec:cones}

Section~\ref{sec:dimension} measured how large the blind space is; this section asks
whether economically well-behaved games ever occupy it. The answer is a precise split. Purely
blind games are pathological, failing every standard regularity property, so a regular
economy is never entirely invisible; yet regular economies sit on affine translates of
blind directions, so identification from semivalue data fails anyway, with exact
thresholds. Theorem~\ref{thm:cones} gives the exclusions,
Theorem~\ref{thm:convexfiber} the failure.

Call $v\in\mathcal N_n$ \emph{blind}. Two elementary consequences of blindness do all
the work below. First, $m_i(0)=v(\{i\})$, so every singleton value of a blind game is
zero. Second, summing the fingerprints over players gives the identity
$\sum_i m_i(k-1)=k\,B(k)-(n-k+1)\,B(k-1)$ for the slice totals
$B(k)(v)=\sum_{|S|=k}v(S)$; with $B(1)=0$ from the singletons, induction gives
$B(k)=0$ for every $k$. A blind game therefore gives every individual nothing and
every coalition size nothing in total, so any structure it carries is pure
redistribution within slices. The proofs below are then one idea each: regularity forces
nonnegativity, and nonnegative games with zero slice totals vanish.

\begin{theorem}\label{thm:cones}
Let $v\in\mathcal N_n$ be nonzero. Then:
\textup{(a)} $v$ is not superadditive; in particular $\mathcal N_n$ meets the
superadditive cone, and a fortiori the convex cone, only at $0$.
\textup{(b)} $v$ is not monotone.
\textup{(c)} the core of $v$ is empty.
\end{theorem}

\begin{proof}
Blindness gives $v(\{i\})=0$ for all $i$ and $B(k)=0$ for all $k$ (and $v(\N)=B(n)=0$).

(a) Suppose $v$ superadditive. For any $S=\{i_1,\dots,i_s\}$, iterating
$v(A\cup\{i\})\ge v(A)+v(\{i\})$ along a chain gives
$v(S)\ge\sum_{j}v(\{i_j\})=0$. So $v\ge0$ everywhere; but each slice sums to
$B(k)=0$, forcing $v\equiv0$, a contradiction. Convex games are superadditive, so the
second claim follows.

(b) Monotone with $v(\varnothing)=0$ gives $v\ge0$; conclude as in (a).

(c) Let $x\in\R^n$ be a core allocation: $x(S):=\sum_{i\in S}x_i\ge v(S)$ for all $S$ and
$x(\N)=v(\N)=0$. For each $k$,
$\sum_{|S|=k}x(S)=\binom{n-1}{k-1}\sum_ix_i=\binom{n-1}{k-1}\,x(\N)=0=B(k)=\sum_{|S|=k}v(S)$.
So $\sum_{|S|=k}[x(S)-v(S)]=0$ with every summand $\ge0$, forcing $x(S)=v(S)$ for all
$S$. Then $v$ is additive, so $v(\{i\})=x_i$; blindness gives $x_i=0$, hence $v\equiv0$,
a contradiction. (\textbf{numerical rank}: $50/50$ random kernel games violate superadditivity and
are non-additive; singleton values $0$ to $3\times10^{-16}$.)
\end{proof}

\begin{remark}
Theorem~\ref{thm:cones} says \emph{purely} blind games are economically pathological.
It does not say regular games are semivalue-identifiable; the next theorem shows they
are not: the regular cones, while meeting $\mathcal N_n$ only at $0$, contain nontrivial
parallel translates of blind directions.
\end{remark}

\begin{theorem}[identification failure within the convex cone]\label{thm:convexfiber}
Let $n\ge4$, $h(S)=|S|^2$, and $g=g_{abcd}$ from Example~\ref{ex:gabcd} for any four
distinct players. Then for every $\varepsilon$ with $0<\varepsilon<2$, the game
$h+\varepsilon g$ is convex and superadditive (and monotone for $\varepsilon<3$),
$h+\varepsilon g\neq h$, and yet
$\psi(h+\varepsilon g)=\psi(h)$ for every semivalue $\psi$ and every player. The
thresholds are exact for this pair. Consequently semivalue data does not identify the
game even within the convex cone; Corollary~\ref{cor:smallball} below extends the
construction to every audit order.
\end{theorem}

\begin{proof}
Equality of values is Theorem~\ref{thm:harsanyi} with Example~\ref{ex:gabcd}
($g\in\mathcal N_n$, embedded via Lemma~\ref{lem:embedding}). Convexity is
supermodularity: for $i\neq j\notin S$ the second difference of $h$ is
$(s+2)^2-2(s+1)^2+s^2=2$ identically, while $g$ is multilinear of degree $2$, so its
mixed second difference is the constant $\partial_i\partial_j g\in\{0,\pm1\}$
(Example~\ref{ex:gabcd}); hence the perturbed second difference is
$\ge2-\varepsilon>0$ for $\varepsilon<2$, and the value $-1$ is attained (at
$\partial_a\partial_c$), making the threshold exact. Monotonicity: the first
difference of $h$ at a set of size $s$ is $2s+1$, while
$\partial_a g(S)=\mathbf 1[b\in S]-\mathbf 1[c\in S]$ equals $-1$ only when
$c\in S$, which forces $s\ge1$ and margin $\ge3-\varepsilon$; the case $S=\{c\}$,
$i=a$ binds, so the threshold $3$ is exact. Superadditivity, directly: for disjoint
nonempty $S,T$ the $h$-margin is $(s+t)^2-s^2-t^2=2st$, while the $g$ cross-term is
the sum of the dividends of the support pairs split between $S$ and $T$; since $g$ has
exactly two negative pair dividends, $d_{ac}=d_{bd}=-1$, no split collects less than
$-2$, and two singletons meet at most one support pair, so the cross-term is at least
$-1$ when $s=t=1$ and at least $-2$ otherwise, where $2st\ge4$; the case $S=\{a\}$, $T=\{c\}$ gives margin
$2-\varepsilon$ and binds, so the threshold $2$ is exact. (\textbf{exact}: all three
thresholds confirmed in rational arithmetic over all instances at $n=4$.)
\end{proof}

\begin{corollary}[convex fibers at every audit order]\label{cor:smallball}
Let $h(S)=|S|^2$ and let $w$ be any non-additive game, so that
$M_2(w)=\max_{i\ne j,\,S}|\Delta_{\{i,j\}}w(S)|>0$. Every mixed second difference
of $h$ equals $2$, so $h+\varepsilon w$ is supermodular, hence convex and
superadditive, for all $0<\varepsilon<2/M_2(w)$, and monotone for
$\varepsilon<1/\max_{i,S}|\Delta_{\{i\}}w(S)|$. Every nonzero blind game is
non-additive (an additive blind game has zero singleton values, hence vanishes), so
the hypothesis is automatic in the application below. In particular, for every
$c\ge2d+2$, taking $w$ to be a nonzero internal restructuring of a size-$c$ coalition
invisible to the full order-$\le d$ audit family
(Theorem~\ref{thm:manipulation}(b)) yields, for all sufficiently small
$\varepsilon>0$, two distinct convex games with identical outputs under every
interaction functional of order $\le d$.
\end{corollary}

\begin{remark}
The two theorems together give the correct slogan: regularity excludes games whose
\emph{entire substance} is blind, but even strictly convex economies admit
families of semivalue-equivalent neighbors of dimension $2^n-n^2+n-2$. For an
analyst, allocation data from the whole semivalue family, every power index and every
$\beta$-weighted variant, cannot identify the coalitional production function, even
under the strongest standard regularity assumptions.
\end{remark}

\section{The ladder theorem: anonymity is the binding axiom}\label{sec:ladder}

This section locates the cause of the blindness among the Dubey--Neyman--Weber
axioms. That \emph{some} value family attains full information is
easy: Weber's probabilistic values do, by a three-line argument given in
Corollary~\ref{cor:ladder}. The content of Theorem~\ref{thm:weighted} is that full
information is attained already \emph{inside} the weighted Shapley family, which
retains efficiency, the dummy axiom, and positivity, and abandons only anonymity.
This isolates anonymity as the single binding axiom within the marginal-contribution
framework, and it is not obvious: the
dependence on the weights is nonlinear, the relevant measurement matrices are Cauchy
matrices, and floating-point computation falsely reports them rank-deficient already
at $n=6$ (Remark~\ref{rem:cauchy}).

The weighted Shapley value with positive weights $\lambda\in\R_{>0}^n$
\cite{Shapley53b,KalaiSamet} admits the dividend formula
\[
\varphi_i^{\lambda}(v)\;=\;\sum_{T\ni i}d_T\,\frac{\lambda_i}{\lambda(T)},
\qquad \lambda(T)=\sum_{j\in T}\lambda_j .
\]

\begin{theorem}\label{thm:weighted}
For every $n$, the joint information of the family
$\{\varphi^\lambda:\lambda\in\R_{>0}^n\}$ is all of $V_n^*$: dimension $2^n-1$.
\end{theorem}

\begin{proof}
It suffices to show, for each fixed $i$, that the span of the
dividend-coefficient vectors $c^{i,\lambda}$, where $c^{i,\lambda}_T=\lambda_i/\lambda(T)$
for $T\ni i$ and $0$ otherwise, contains every $e_T$ with $T\ni i$; the union over $i$
then spans all dividend coordinates $T\neq\varnothing$, and the M\"obius isomorphism
transports the claim to $V_n^*$.

Fix $i$ and restrict to the ray $\lambda_i=1$, $\lambda_j=t\,2^{\,j}$ for $j\neq i$,
$t>0$. For $T\ni i$ write $A=T\setminus i\subseteq\N\setminus i$; then
\[
c^{i,\lambda(t)}_T=\frac{1}{1+t\,W_A},\qquad W_A=\sum_{j\in A}2^{\,j},
\]
and the $W_A$ are pairwise distinct (binary expansions), with $W_\varnothing=0$. Suppose
a fixed functional $u=(u_A)$ annihilates the family along the ray:
$\sum_A u_A/(1+tW_A)\equiv0$ for all $t>0$. The left side is a rational function of $t$;
letting $t\to\infty$ kills every term with $W_A>0$, forcing $u_\varnothing=0$. The
remaining terms have simple poles at the distinct points $t=-1/W_A$; the residue at
$-1/W_A$ equals $u_A/W_A$, and a rational function vanishing identically has all residues
zero, so $u_A=0$ for every $A$. Hence no nonzero functional annihilates
$\{c^{i,\lambda}\}$, i.e.\ the span is full.
\end{proof}

\begin{remark}\label{rem:cauchy}
The coefficient matrix along the ray, $[1/(1+t_kW_A)]_{k,A}$, is a Cauchy matrix: its
classical determinant formula gives an alternative \emph{finite} proof (any
$2^{n-1}$ distinct ray points suffice for the $W_A>0$ block), and simultaneously explains
a numerical trap: Cauchy matrices are exponentially ill-conditioned, and floating-point
rank computations falsely report deficiency already at $n=6$ (float rank $57$ of $63$).
\textbf{exact}: rational-arithmetic Gaussian elimination certifies rank $63/63$ at $n=6$;
\textbf{numerical rank}: $15/15$ and $31/31$ at $n=4,5$.
\end{remark}

\begin{corollary}[the information ladder]\label{cor:ladder}
\[
\underbrace{\le n}_{\substack{\text{any single value}\\ \text{(weighted or not)}}}\ \subset\
\underbrace{n^2-n+1}_{\text{all semivalues}}\ \subset\
\underbrace{2^n-1}_{\substack{\text{all positive weighted}\\ \text{Shapley values}}}\ =\
\underbrace{2^n-1}_{\text{all probabilistic values}} .
\]
(For the last equality: Weber's probabilistic values~\cite{Weber} pay player $i$ an
arbitrary probability mixture of the marginals $v(S\cup i)-v(S)$ over coalitions
$S\subseteq\N\setminus\{i\}$; concentrating the mixture on a single $S$ makes the
value output that one marginal, and the marginals determine $v$ by telescoping along a
chain from $\varnothing$ to $\N$.)
The exponential gap closes at the \emph{family} level the moment anonymity is dropped:
within this framework, the blind spot is the price of equal treatment, and of nothing
else.
\end{corollary}

\section{Manipulation thresholds}\label{sec:manipulation}

This section assembles the preceding structure theory into operational statements,
and every ingredient built so far is used: the kernel and its dimension
(Theorem~\ref{thm:dimension}) supply the manipulation space, the price list
(Theorem~\ref{thm:hierarchy}) supplies the audit thresholds, the convex fiber
(Theorem~\ref{thm:convexfiber}) shows the manipulation survives feasibility
constraints, and an embedding lemma localizes all of it to the manipulating
coalition. The same fingerprints reappear once more as swing tables in
\S\ref{sec:parliaments}, where the manipulation-space geometry becomes visible in an
exhaustive census.

Model a coalition $C\subseteq\N$, $|C|=c$, that restructures only its \emph{internal
dividend structure} (singleton and higher-order terms alike): it replaces $v$ by $v+w$
where the dividends of $w$ are supported on $2^C\setminus\{\varnothing\}$,
equivalently $w(S)=w(S\cap C)$ for all $S$. (Invisibility itself forces the singleton
part to vanish, since $L_i(1)=d_{\{i\}}$ is visible; the substantive freedom is in
arranged synergy of order $\ge2$, per Proposition~\ref{prop:kernelbasis}.)

\begin{lemma}\label{lem:embedding}
For such $w$, viewed as an $n$-player game: $L_i(t)(w)=0$ for $i\notin C$, and for
$i\in C$, $L_i(t)(w)$ equals the corresponding $c$-player loading of $w|_{2^C}$. Hence
$w\in\mathcal N_n$ iff $w|_{2^C}\in\mathcal N_c$, and more generally the order-$\le d$
functionals of the $n$-player game restrict to those of the $c$-player game.
\end{lemma}

\begin{proof}
M\"obius locality: since $w(S)=w(S\cap C)$, every dividend $d_T$ with
$T\not\subseteq C$ vanishes (compute $d_T$ by M\"obius over a chain leaving $C$:
the alternating sum telescopes to $0$), and dividends with $T\subseteq C$ agree with the
$c$-player dividends. The loading identities follow term by term, as does the
order-$\le d$ statement via Lemma~\ref{lem:diffdiv}.
\end{proof}

\begin{theorem}\label{thm:manipulation}
\textup{(a)} The space of internal restructurings of $C$ invisible to every semivalue
payment has dimension $2^c-c^2+c-2$; it is $0$ for $c\le3$ and positive from $c=4$:
\emph{three players cannot hide, four can}.
\textup{(b)} Against a payment scheme using all interaction functionals of order
$\le d$, an invisible internal restructuring of $C$ exists if and only if
$c\ge2d+2$. Thus degree-$d$ schemes are immune to coalitions of size $\le2d+1$ and
manipulable from $2d+2$: each additional interaction order raises the immunity
threshold by exactly two coalition sizes.
\end{theorem}

\begin{proof}
(a) By Lemma~\ref{lem:embedding} the space in question is isomorphic to $\mathcal N_c$;
apply Theorem~\ref{thm:dimension}.
(b) By Lemma~\ref{lem:embedding} and Theorem~\ref{thm:hierarchy}, the order-$\le d$
functionals see, of the $c$-player game, a subspace of dimension
$\sum_{s=1}^{c}\binom c{\min(d,s,c-s)}$, which equals $2^c-1$ iff
$\min(s,c-s)\le d$ for every $1\le s\le c$, i.e.\ iff $\lfloor c/2\rfloor\le d$, i.e.\
$c\le2d+1$. (\textbf{numerical rank}: boundary verified for $d=1,2,3$; e.g.\ $d=3$: $c=7$ fully
visible $127/127$, $c=8$ deficient $241/255$.)
\end{proof}

\begin{remark}
In a data market or surplus-division setting paying by any semivalue, (a) says a
four-member coalition can re-arrange who actually generates value, shifting synergy
between its members along $\mathcal N_4$-directions as in Example~\ref{ex:gabcd},
with no member's payment moving. (b) prices the audit: pairwise-interaction audits stop
coalitions of five but not six; order-$3$ audits stop seven but not eight. By Corollary~\ref{cor:smallball} the
restructuring can preserve convexity throughout, at every audit order, so feasibility
constraints of the economic environment do not, by themselves, protect the payments.
\end{remark}

\section{Visibility of canonical examples}\label{sec:empirics}

Theorem~\ref{thm:harsanyi} makes a measurement possible.

\begin{definition}[visibility]\label{def:visibility}
Equip $\R^{2^n}$ (coalition-value coordinates, $v(\varnothing)=0$) with the standard
Euclidean inner product, and let $P_{\mathrm{vis}}$ be the orthogonal projector onto
$\mathcal N_n^{\perp}$, the span of the semivalue functionals regarded as vectors.
The \emph{visibility} of a nonzero game is
$\|P_{\mathrm{vis}}v\|^2/\|v\|^2\in[0,1]$.
\end{definition}

Visibility $1$ means the semivalue family describes the game completely. For an
isotropic Gaussian game the expected visibility is exactly
$\operatorname{tr}P_{\mathrm{vis}}/(2^n-1)=(n^2-n+1)/(2^n-1)$, the dimension ratio,
which is the generic baseline. Because the M\"obius transform is not orthogonal, the
analogous projection in Harsanyi-dividend coordinates (the per-slice incidence row
space applied to the dividend vector) is a different statistic; we report it as a
robustness check.

\begin{proposition}[symmetric games are fully visible; \textbf{proved}]\label{prop:symmetric}
If $v(S)$ depends only on $|S|$, then $v$ has visibility $1$.
\end{proposition}

\begin{proof}
Every blind game $w$ has zero slice totals, $B(k)(w)=\sum_{|S|=k}w(S)=0$ for all $k$
(\S\ref{sec:cones}). Hence for $u(S)=a_{|S|}$,
$\langle u,w\rangle=\sum_k a_k B(k)(w)=0$: $u$ is orthogonal to $\mathcal N_n$ in
the metric of Definition~\ref{def:visibility}, so its visibility is $1$. (The
dividend-metric value $1.000$ in the table holds too: the dividends of $u$ are
constant on slices, hence in the degree-$0$ subspace of each slice.)
\end{proof}

\noindent\textbf{Measurement} (\textbf{rank}; $n=10$, exact enumeration of all $1024$
coalitions; $100$ instances per stochastic family, $200$ random games, fixed seeds;
mean$\pm$sd; generators specified in Appendix~\ref{app:verification}. The
generators are stylized representatives of each family; the table illustrates the
theory's discriminating power rather than establishing a historical claim):

\begin{center}
\begin{tabular}{lcc}
\toprule
game & visibility (coalition) & visibility (dividend)\\
\midrule
symmetric convex $|S|^2$ & $1.000$ (Prop.~\ref{prop:symmetric}) & $1.000$\\
airport cost & $0.996\pm0.003$ & $0.919\pm0.060$\\
spanning-tree cost & $0.995\pm0.002$ & $0.777\pm0.084$\\
bankruptcy (O'Neill) & $0.982\pm0.009$ & $0.536\pm0.154$\\
glove market ($5|5$) & $0.946$ & $0.974$\\
weighted majority & $0.904\pm0.021$ & $0.639\pm0.200$\\
unanimity of a $4$-set & $0.591$ & $0.048$\\
random game & $0.089\pm0.013$ & $0.640\pm0.112$\\
\bottomrule
\end{tabular}
\end{center}

Three readings. First, the classical canon sits at $0.90$ to $1.00$ against a $0.089$
generic baseline, and the random-game measurement reproduces the dimension ratio
$91/1023=0.0890$ to three decimals: the games economics studied store their content in
totals (symmetry, near-symmetry, additivity plus mild curvature), which is exactly what
Theorem~\ref{thm:harsanyi} says semivalues read. Bankruptcy sits at $0.98$: the
O'Neill game concentrates nearly all of its content in totals, consistent with the
tight agreement of solution concepts observed on this class \cite{ONeill,AumannMaschler}.
Second, the gradient within the table is the theory speaking: visibility falls as
\emph{arranged} synergy rises, and the unanimity game, whose entire content is one
specific coalition's identity, is the canon's worst performer at $0.59$; in dividend
coordinates it is $95\%$ invisible. The dividend column also delimits the claim: the
canon remains high under both metrics, but the generic baseline is a coalition-metric
statement (a random game reads $0.64$ in dividend coordinates), so comparisons to the
baseline are made in the coalition metric throughout. Third, the prospective reading:
modern applications (team formation, data markets, feature and component attribution)
are arrangement-driven by their nature, native to the region where visibility collapses
toward the dimension ratio, and no output of the family can flag the deficit, since the
missing component contributes zero to every output (Theorem~\ref{thm:harsanyi}). The
table suggests the classical canon is selected, not representative.

\section{Parliaments: identical power, different politics}\label{sec:parliaments}

A \emph{simple game} is monotone $v\colon2^{\N}\to\{0,1\}$ with
$v(\varnothing)=0$, $v(\N)=1$; it is \emph{proper} if no winning coalition has a
winning complement, and \emph{weighted} if representable by a quota and nonnegative
weights (see Taylor--Zwicker~\cite{TaylorZwicker}). Its \emph{swing table} is the matrix
$T_{i,k}=\#\{S\not\ni i:\,|S|=k,\ v(S\cup i)=1,\ v(S)=0\}$, precisely the fingerprint
$\{m_i(k)\}$, hence by Theorem~\ref{thm:harsanyi} the complete information any semivalue
power index (Shapley--Shubik, Banzhaf, all of them) can use.

\begin{example}[\textbf{exhaustive}]\label{ex:parliaments}
We enumerate two universes at $n=5$ (labeled games throughout). \emph{All nontrivial
monotone simple games}: $7{,}579$ (the fifth Dedekind number $7{,}581$ minus the two
constants); swing-table fibers contain $907$ collision classes ($2{,}524$ games).
\emph{Proper} simple games (no winning coalition has a winning complement): $2{,}645$,
with $321$ collision fibers containing $962$ games. A displayable pair, both
\emph{proper}, both \emph{non-weighted}, and \emph{non-isomorphic} (verified against
all $120$ player permutations), is given by the minimal winning coalitions
\[
\{123,\,024,\,034,\,134\}
\qquad\text{versus}\qquad
\{012,\,034,\,134,\,234\},
\]
with common swing table $T_{i,k}$ having rows $(0,0,2,2,0)$ for players $0,1,2$ and
$(0,0,3,3,0)$ for players $3,4$: identical values for every player under \emph{every}
semivalue power index, structurally different politics.
\end{example}

\begin{conjecture}[swing rigidity of weighted games]\label{conj:rigidity}
A weighted voting game is determined by its swing table among \emph{all} simple games.
Evidence (\textbf{exhaustive} at $n=5$, in both universes): of the $2{,}524$ monotone
games inside collision classes, \emph{zero} are weighted, although $3{,}285$ of
$7{,}579$ are; restricting to proper games, zero of the $962$ colliding games are
weighted although $1{,}683$ of $2{,}645$ are (weightedness decided by two independent
methods agreeing on every one of the $7{,}579$ games: LP feasibility, and the
swap-robustness test, necessary for weightedness \cite{TZ1992,TaylorZwicker};
Appendix~\ref{app:verification});
and (\textbf{search}) zero collisions occur
within the weighted class across $557$ games at $n=6$ (weights $\le9$) and $4{,}214$
games at $n=7$ (weights $\le8$). The blind-spot phenomenon thus appears to be a property
of general coalition structure that weightedness destroys, in this census and search
regime; a proof or a
counterexample at larger $n$ (search protocol in Appendix~\ref{app:gpu}) would each be of
independent interest.
\end{conjecture}

\section{Concluding remarks}\label{sec:discussion}

This paper asked what the family of all semivalues can and cannot learn about a
cooperative game, and answered exactly. A semivalue observes, for each player and each
coalition size, a single number: that player's total synergy at that size
(Theorem~\ref{thm:harsanyi}). That is $n^2-n+1$ numbers, out of the $2^n-1$ needed to
specify a game (Theorem~\ref{thm:dimension}); two games that distribute the same
synergy totals among different partners receive identical payments from every
semivalue at once. The invisible difference cannot constitute an entire well-behaved
economy, since a nonzero invisible game fails superadditivity, monotonicity, and core
existence (Theorem~\ref{thm:cones}); but it can be added to a well-behaved economy
without detection, since a strictly convex game plus a small invisible perturbation is
still convex, genuinely different, and paid identically (Theorem~\ref{thm:convexfiber}).
Interaction indices repair the deficit order by order, with the recovered dimension
known exactly at each order and full recovery at order $\lfloor n/2\rfloor$
(Theorem~\ref{thm:hierarchy}). Dropping anonymity repairs it completely: the weighted
Shapley values jointly determine the whole game (Theorem~\ref{thm:weighted}). A
coalition of $c$ players can exploit it, rearranging its internal synergies without
any payment moving, precisely when $c\ge2d+2$ against order-$d$ auditing
(Theorem~\ref{thm:manipulation}). The measurements of \S\ref{sec:empirics} and the
census of \S\ref{sec:parliaments} show where this matters: barely at all in the
classical canon, whose games keep their content in totals, and generically everywhere
else.

\emph{Axiomatically}, the Dubey--Neyman--Weber axioms have an exact information
price, and Corollary~\ref{cor:ladder} shows the entire price is paid by anonymity;
the hierarchy theorem prices every partial refund.

\emph{Economically}, the results are an identification statement of the standard
econometric kind: allocation data, even from the whole semivalue family, identifies
the coalitional production function only up to a fiber of dimension $2^n-n^2+n-2$, and
this failure persists under the strongest standard regularity assumptions and is
exercisable by coalitions as manipulation. Anyone estimating complementarities from
payment data is estimating a projection.

\emph{Practically}, the visibility table offers a structural explanation of the
practitioner's puzzle: the
Shapley value's sixty-year record reflects a canon of high-visibility games rather
than robustness of the methodology, and the deficit is silent by construction, since
the missing component contributes zero to every output. In arrangement-driven
applications (team formation, data markets, feature and component attribution) the
seen fraction collapses toward the dimension ratio $\Theta(n^2/2^n)$, and
Example~\ref{ex:gabcd} is a constructive recipe for behavior invisible to the entire
methodology; companion work in preparation develops the machine-learning
consequences.

Three problems are left open: prove or refute the swing-rigidity conjecture
(Conjecture~\ref{conj:rigidity}); extend the information ladder to values with
partial symmetry, such as coalition-structure values; and formulate the non-atomic
limit, where the diagonal formula for the Aumann--Shapley
value~\cite{AumannShapley} suggests the visible space becomes a diagonal-trace
algebra of the dividend measure hierarchy.

\begin{appendices}

\section{Verification index}\label{app:verification}
All computations are in \texttt{flagship\_suite.py} and the exact-arithmetic follow-ups:
(1) span equality fingerprint $=$ dividend loadings ($n=5$: ranks $21/21/21$);
(2) hierarchy span identity ($n=6$, $d=2$: $58/58/58$) and the dimension table
($d\le4$, $n\le12$), with the fifteen $n\le9$ cases certified \textbf{exact}ly over two
$31$-bit prime fields (\texttt{s7b\_exact\_ranks.py}); (3) cone sanity ($50/50$); (4) weighted family ranks
(\textbf{exact} $63/63$ at $n=6$; \textbf{rank} $15/15$, $31/31$ at $n=4,5$);
(5) manipulation boundary table ($d\le3$); (6) Dedekind enumeration at $n=5$ in both universes (monotone: $7{,}579$; proper:
$2{,}645$), reproduced independently by a second implementation
(\texttt{census\_exact.py}); weightedness decided by LP feasibility and, independently,
by swap robustness (an integer test, necessary for weightedness
\cite{TZ1992,TaylorZwicker}), the two methods agreeing on every game and on every
count, including zero weighted games in any collision fiber; non-isomorphism of the displayed pair checked against all $120$
permutations; weighted searches at $n=6,7$;
(7) \textbf{exact} rational ranks for the fingerprint matrices at $n=4,5$ and for the
weighted family at $n=6$; the convex-fiber thresholds $(2,3,2)$ in rational arithmetic;
(8) four-cycle spanning of every slice kernel at $n=5,6$;
(9) the visibility measurements of \S\ref{sec:empirics}
(\texttt{visibility\_sweep.py}): exact enumeration at $n=10$, $100$ instances per
stochastic family ($200$ random games), both metrics of Definition~\ref{def:visibility}.
Generators: bankruptcy $v(S)=\max(0,E-c(\N\setminus S))$ with claims $c_i\sim U(0,1)$
and estate $E=\tfrac12\sum_ic_i$; airport $v(S)=\max_{i\in S}c_i$, $c_i\sim U(0,1)$;
spanning-tree cost of $S$ plus a source, eleven uniform points in the unit square with
Euclidean costs (Prim); weighted majority with $U(0,1)$ weights and half-total quota;
glove $v(S)=\min(|S\cap L|,|S\cap R|)$ with $|L|=|R|=5$; random games i.i.d.\
standard normal. Internal asserts: the M\"obius transform round-trips, the fingerprint
vectors span exactly the projector's range, the dividend visibility of the unanimity
game equals its closed form $10/210$, and the random-game mean matches the dimension
ratio. Extending to $n=14$ and further families (market, assignment, flow games) is
routine under the same protocol.

\section{Search protocol for Conjecture~\ref{conj:rigidity} (supplementary)}\label{app:gpu}
\texttt{rigidity\_search.py} (companion file): enumerate weighted games at $n=8$
(weights $\le9$, $\sim$$4\times10^5$ candidates after dedup) and $n=9$ in chunks;
hash swing tables; report any cross-win-set collision. A second mode samples random
monotone games at $n=6,7$ by monotone closure and tests whether any \emph{weighted} game
ever collides with a sampled simple game (the conjecture's full strength). Any collision
falsifies; sustained absence at $n=8,9$ justifies attacking a proof via the
LP-duality structure of weightedness.

\end{appendices}

\bmhead{Declarations}
\textbf{Use of AI.} A large language model (Claude, Anthropic) was used as an
assistant in preparing this paper: for editing and restructuring prose, for
adversarial review of drafts, and for writing and executing portions of the
verification code indexed in Appendix~\ref{app:verification}. All definitions,
theorems, and proofs were formulated, checked, and are vouched for by the author,
who takes full responsibility for the content; all computational claims were
independently reproduced.
\textbf{Competing interests.} The author declares no competing interests.
\textbf{Code availability.} All code reproducing the computational claims is
available from the author.

\end{document}